\documentclass[journal,twocolumn]{IEEEtran}

\pdfoutput=1
\usepackage{hyperref}
\usepackage{pifont}
\usepackage{graphics}
\usepackage{graphicx}

\newcommand\old[1]{{\color{gray} #1}}
\usepackage[utf8]{inputenc}
\usepackage[utf8]{inputenc}
\usepackage[english]{babel}
 
\usepackage{amsthm}
\usepackage{siunitx}
\usepackage[svgnames,rgb]{xcolor}
\usepackage{amsmath,amssymb}
\usepackage{tikz}
\usetikzlibrary{arrows,patterns}
\usepackage{cite}
\usepackage{multirow}
\usepackage[ruled]{algorithm2e}
\usepackage{xparse}% http://ctan.org/pkg/xparse
\usepackage{mathrsfs} 
\usepackage[mathscr]{euscript} 
\DeclareMathAlphabet{\mathpzc}{T1}{pzc}{m}{it}
\usepackage{mathrsfs}

\newtheorem{Proposition}{Proposition}

\usepackage{pgfplots}
\pgfplotsset{compat=newest}% <-- moves axis labels near ticklabels (respects tick label widths)

\newtheorem{theorem}{\textbf{Theorem}}

\newtheorem{remark}{\textbf{Remark}}

\newcounter{MYtempeqncnt}

\renewcommand\old[1]{}

\usepackage[utf8]{inputenc}

\title{CDC 2018}
\author{Amir Khazraei}
\date{November 2018}

\begin{document}
\title{\LARGE \bf An Optimal Linear Dynamic Detection Method for
Replay Attack in Cyber-Physical Systems}

\author{Amir Khazraei, Hamed Kebriaei$^*$, and Farzad Rajaei Salmasi 
\thanks{This research was in part supported by a grant from the Institute for Research in Fundamental Sciences (IPM) (No. CS 1397-4-56).}
\thanks{Authors are with the School of Electrical and Computer Engineering,
 College of Engineering, University of Tehran, Tehran, Iran. H. Kebriaei is also with School of Computer Science, Institute for Research in Fundamental Sciences (IPM), P.O. Box 19395-5746, Tehran, Iran.}
\thanks{$^*$ Corresponding author: H. Kebriaei (e-mail: kebriaei@ut.ac.ir). }
 }

\maketitle \thispagestyle{empty} \pagestyle{empty}

%%%%%%%%%%%%%%%%%%%%%%%%%%%%%%%%%%%%%%%%%%%%%%%%%%%%%%%%%%%%%%%%%%%%%%%%%%%%%
\begin{abstract}
The problem of detecting replay attack to the linear stochastic system with Kalman filer state estimator and LQG controller is addressed. To this end, a dynamic attack detector method is proposed which is coupled with
the dynamics of the system. While preserving stability of the main system, conditions on parameters of the attack detector dynamics are obtained such that the attack can be revealed by destabilization of a residual trajectory which is the difference between the estimated and measured output of the system. Using this method, system operator can adjust the detection rate based on the proposed scheme by changing the design parameters. Nevertheless, since the exogenous attack detector signal affects the performance of the closed loop control system, we propose an optimization problem to determine such a detector with minimum loss effect. In the simulation results, the proposed dynamical attack detector approach is compared with the well-known additive white noise watermarking method and the results confirm the superiority of the new scheme. 
\end{abstract}

\maketitle

\section{Introduction}\label{sec1}

Cyber-physical systems refer to a new generation of systems which integrates computational, networking and physical processes. Security of Cyber-Physical systems from control theory point of view has gained lots of attention during the last few years [1]. Due to widespread use of computer network in control of physical systems, the malicious entities have seized an opportunity to inject attacks on cyber-physical systems [2-3]. 

The analysis over vulnerabilities of cyber-physical systems has attracted a tremendous amount of attention in past few years. For instance, false data injection attack against static state estimation in a power grid has been studied in [4]. The authors showed that if the injected signal to the sensory measurements of the system lays within the range space of the observation matrix, then the attack remains stealthy. Conditions on the existence of undetectable and unidentifiable attacks has been characterized in [5] using graph theoretic method. It is assumed that the adversary is capable of injecting the attack signal into the input and output of the system. In [6], the cyber-physical attack is modeled based on adversary resources including disruption, disclosure and information of the system. Covert misappropriation attack against a networked control system has been considered in [7]. In this type of attack, attacker needs the complete information of the system and more resources compared to other types of attack. Other standard types of attacks like Denial of Service (DoS) and deception attacks have been considered in [8-10].

The problem of resilient state estimation in a standard form has been addressed in [11-14]. They all assume that the attacker injects attack vector to less than half of sensors, since it has been shown that it is impossible to reconstruct the state of the system whenever more than half of the sensors are under attack. Reconstructing the state of a nonlinear control system subject to sensor and actuator attacks has also been studied in [15]. However, it is assumed that some set of sensors are trustworthy. But, in the problem of replay attack detection addressed in this paper, we assume that all sensors can be compromised. It should be noted that our goal is only to detect the presence of attack rather than estimating the true state of the system during the attack time.

Replay attack against a centralized control system has been introduced in [16] and [17] where for discrete time linear time invariant (LTI) control systems an optimal noisy watermarking signal (independent and identically distributed noise with zero mean and specific constant covariance) has been proposed in order to detect the attack. In a recent article of the same authors [18], a stationary process is utilized such that the defender has more degree of freedom to design the watermarking signal for a discrete time control system, which improves the detection rate with the same performance loss. The similar iid approach has been proposed in [19] while the dynamic of the system is unknown to the defender.

Detection of replay attack in frequency domain has been introduced in [20]. It is claimed that the attack can be detected without injecting watermarking signal. However, a restricting assumption is made so that a white Gaussian noise signal with enough power should exist into the communication network system to detect the replay attack. Moreover, it has been shown in [21] that by including a nonlinear term in the control system in non-regular time intervals, the attack can be detected.  But, the performance loss induced by incorporating such a nonlinear term into the system is not analyzed. The problem of replay attack against multi agent systems has been studied in [22] and a method based on sharing white Gaussian watermarking signal between agents has been proposed. The authors also initiated the idea of dynamic attack detector signal for the replay attack to the linear continuous time systems in [23].

In this paper, we introduce a new method based on injection of dynamical attack detector signal to detect the replay attack against discrete time LTI stochastic systems. In the replay attack, it is assumed that there exists an adversary who is capable of recording the sensory measurements for a period of time, injecting a disruptive signal into the system and altering the real sensory measurements of the system with the recorded measurements. It is also assumed that the defender has the full information of the dynamic of the system. We provide the conditions on the design parameters of the attack detector signal so that the replay attack is detected while the system  remains stable in the absent of attack. 
Our approach yields the following advantages and contributions with respect to the literature

\begin{enumerate}
 \item The proposed method in [16] and [17] are based on injection of the white and colored noise to the system, respectively that render a \textit{constant} probability of attack detection. But in our method, the attack detector signal is generated based on a \textit{dynamic} that uses output estimation error. The probability of attack detection in our method raises as the time increases and the attack can be detected in a shorter period of time.
 \item We provide the condition under which the attack can be revealed using the proposed  attack detector signal (Theorem 1). It is also proved that by applying the proposed attack detector signal, it is always possible for physical systems to meet such condition (Theorem 2, proposition 3 and Remark 7). 

\item The detection rate of replay attack can be adjusted by design parameters in our method and an optimization problem is proposed so that the performance loss of the system after injection of the attack detector signal is minimized subject to detection rate constraint.
\item  The performance of the proposed method is assessed from two aspects
  
 \begin{itemize}
\item 	It has been shown in simulation results that our proposed method is more useful when defender wants to detect the attack with greater level of accuracy, since within a finite time, our method can render more number of alarms in an arbitrary window time.
\item 	Since we do not inject the noise directly to the system, the control signal in our proposed attack detector has smaller range and lower energy with respect to the control signal in white Gaussian noise method for the same detection time and the same desired number of alarms.
 \end{itemize}

\end{enumerate}

\section{PRELIMINARIES}\label{sec:PRELIMINARIES}

\subsection{System Model}
Consider a control system, described by the following time invariant linear system\\
\begin{equation}\label{eq(1)}
\begin{array}{l}
x(t + 1) = Ax(t) + Bu(t) + w(t)\\
y(t) = Cx(t) + v(t)
\end{array}\end{equation}
where $x \in {\mathbb{R}^n}$, $u \in {\mathbb{R}^m}$, $y \in {\mathbb{R}^p}$ are the state, input and output vectors, respectively. Moreover, $w \in {\mathbb{R}^n}$ and $v \in {\mathbb{R}^p}$ are the process and measurement noises which are assumed to be independent white Gaussian processes with covariance matrices $Q$ and $R$, respectively. Initial condition $x(0)$ is a Gaussian random variable with mean $\mu$ and covariance $\Sigma(0)$ (i.e. $x(0) \sim \mathcal{N}\big(\mu,\,\Sigma(0)\big)$) is independent from process and measurement noises. 
\subsection{Controller Structure}
We assume that the pair $(A,B)$ is controllable and $(A,C)$ is observable. Therefore, one can design a LQG controller that minimizes the following cost
\begin{equation}\label{eq{2}}
J = \mathop {\lim }\limits_{N \to \infty } \frac{1}{N}\sum\limits_{t = 1}^N {E\left\{ {{x^T}(t)Wx(t) + {u^T}(t)Uu(t)} \right\}}\end{equation}
where $W$ and $U$ are semi positive definite and positive definite weight matrices with suitable dimensions, respectively.
For this system, state estimation unit based on the steady state Kalman filter is in the following form
\begin{equation}
\begin{split}
\hat x(t) &= \hat x(t|t - 1) + L\big(y(t) - C\hat x(t|t - 1)\big)\\
\hat x(t + 1|t) &= A\hat x(t) + Bu(t)
\end{split}
\end{equation}
In which $L = {\Sigma _e}{C^T}{(C{\Sigma _e}{C^T} + R)^{ - 1}}$ and ${\Sigma _e}$ is stabilizing solution of the following Riccati equation
\begin{equation}
{\Sigma _e} = A{\Sigma _e}{A^T} - A{\Sigma _e}{C^T}{(C{\Sigma _e}{C^T} + R)^{ - 1}}C{\Sigma _e}{A^T} + Q
\end{equation}                   
and the control input to the system is the LQR gain multiplied by the estimated state of the system as follows
\begin{equation}
u(t) = K\hat x(t)
\end{equation}
where $K =  - {({B^T}PB + U)^{ - 1}}{B^T}PA$  and $P$ is obtained from the following Riccati equation
\begin{equation}
P = {A^T}PA + W - {A^T}PB{({B^T}PB + U)^{ - 1}}{B^T}PA
\end{equation}                                          
%Now by defining $e(t) = x(t) - \hat x(t)$ and $\bar x(t) = \left[ {\begin{array}{*{20}{c}}
%{x(t)}\\
%{e(t)}
%\end{array}} \right]$, we get
%\begin{equation}
%\begin{split}
%\bar x(t + 1) =& \underbrace {\left[ {\begin{array}{*{20}{c}}
%{A + BK}&{ - BK}\\
%0&{A - LCA}
%\end{array}} \right]}_{\bar A}\bar x(t) \\
% &+ \left[ {\begin{array}{*{20}{c}}
%{w(t)}\\
%{(I - LC)w(t) - L\upsilon (t + 1)}
%\end{array}} \right]
%\end{split}
%\end{equation}
%The optimal cost of the LQG controller can be found by
%\begin{equation}
%{J^ * } = tr\left( {\Pi {\Sigma _{\bar x}}} \right)
%\end{equation}
%where
%\begin{equation*}
%\Pi  = \left[ {\begin{array}{*{20}{c}}
%{W + {K^T}UK}&{ - {K^T}UK}\\
%{ - {K^T}UK}&{{K^T}UK}
%\end{array}} \right]
%\end{equation*}
%and ${\Sigma _{\bar x}}$  is the covariance matrix of $\bar x$ that is the %solution of the following Lyapunov equation
%\begin{equation}
%\begin{split}
%{\Sigma _{\bar x}} =& {{\bar A}^T}{\Sigma _{\bar x}}\bar A  \\
%&+\left[ {\begin{array}{*{20}{c}}
%Q&{Q{{(I - LC)}^T}}\\
%{(I - LC)Q}&{(I - LC)Q{{(I - LC)}^T} + LR{L^T}}
%\end{array}} \right]
%\end{split}
%\end{equation} 

\section{PROBLEM FORMULATION}
\subsection{Attack Model}
We consider there exists an adversary that disrupts the control system (\ref{eq(1)}) using replay attack. A replay attack model is defined in [16], where the attacker targets sensory measurements of the system and meanwhile injects a destructive signal into the system.
 The attack procedure is performed in two phases
\begin{enumerate}
 \item During the first phase, from time $-\tau$ to $-1$, attacker records the sensory measurements.
  \item In the second phase, from time zero to $\tau-1$, attacker replaces the sensory measurements of the system with the recorded ones, as follows
 
 \begin{equation*}
 y(t) = y(t - \tau )\,\,\,\,\,\,for\,\,0 \le t \le \tau-1
 \end{equation*}
 In addition, during the second phase of the attack, the attacker injects a disruptive signal ${f}(t)$ through the external input channel $B^a$ to the system as follows
 
 \begin{equation}\label{eq{98}}
 {\begin{array}{*{20}{c}}
 {x(t + 1) = Ax(t) + Bu(t) + B^{a}f(t)}\\
 {y(t) = y(t - \tau ) = {y^\tau }(t)\,\,}
 \end{array}}
 \end{equation} 
\end{enumerate}
 
The signal $f(t)$ is chosen by the attacker and this information is not commonly available to the defender. It is worth mentioning that attacker does not need to know about the model of the system. Having access to disclosure resources in output channel and disruption resources in input and output channels is sufficient to perform the replay attack.\\
We define $r(t) = y(t) - C\hat x(t|t - 1)$ as the residue signal of the system. When the system works on normal condition, the residue signal is an innovation sequence with zero mean and covariance $H = C{\Sigma _e}{C^T} + R$ [24].
Let  define the following weighted norm of the residual signal
\begin{equation}\label{eq{12}}
g(t) = {r^T}(t){{H}^{ - 1}}r(t)
\end{equation}                                                              
which has the Chi-square distribution form. It should be considered that the defender utilizes $g(t)$ to detect the presence of abnormality in the system, by comparing $g(t)$ with a scalar $\eta$ as a threshold which is determined by defender. Anomaly detector will warn the system via an alarm if $g(t) > \eta $, while the system is considered in its normal operation if $g(t) < \eta $. When the system is in its normal condition, we define $\alpha  = P(g(t) > \eta )$ as the false alarm rate [25]. To reduce the number of false alarm rates in a period of time, usually $\eta$ is adjusted in a way that $\alpha$ becomes close to zero.\\
 It is also assumed that there exists a window time $\mathscr{T}$ that defender counts the number of alarms in the interval $(t-\mathscr{T},t)$ at each time step $t$. If the number of alarms in a window time exceeds a threshold (denoted by ${\beta}$) which is determined by the defender, it is considered that the attack is occurred at time $t$. Since the expected number of false alarms in a window time is equal to $\alpha \mathscr{T}$, in order to avoid false attack detection in normal condition of the system, the desired number of alarms in a window time should be chosen enough larger than $\alpha \mathscr{T}$ (i.e. $\alpha \mathscr{T}\ll\beta$). Moreover, we define the detection time as the required time to reach into the desired number of alarms ${\beta}$. \\ 
\textbf{Notation}: In rest of the paper the superscript $\tau$ shows the delayed signal. For instance\\
${x^\tau }(t) = x(t - \tau )$,\\
and superscript $a$ shows the under attack system's signals. Therefore, from (\ref{eq{98}}) we get $y^a(t)=y^\tau(t)$.
 
\begin{remark}
 The amount of $\tau$ is determined by the attacker, and as much as it becomes larger, the attacker will get enough time to disrupt the system and also the defender will have enough time to detect the attack. However, in the case that $\tau$ is too small, we may not be able to detect the attack for a given $\beta$. In fact, the attack is detected if the detection time be lower than $\tau$ for a given $\beta$.
\end{remark}
The system from $ - \tau $ to zero can be described in the following form
\begin{equation}
\begin{split}
{x^\tau }(t + 1) &= A{x^\tau }(t) + B{u^\tau }(t) + {w^\tau }(t)\\
{y^\tau }(t) &= C{x^\tau }(t) + {\upsilon ^\tau }(t)
\end{split}
\end{equation}
where $u^\tau (t) = K{\hat x^\tau }(t)$ and the state estimator of the system in the absence of attack is as follow
\begin{equation}\label{eq{15}}
{{\hat x}^\tau }(t + 1|t) = A{{\hat x}^\tau }(t) + B{u^\tau }(t)
\end{equation}
\begin{equation}\label{eq{16}}
{{\hat x}^\tau }(t) = {{\hat x}^\tau }(t|t - 1) + L[{y^\tau }(t) - C{{\hat x}^\tau }(t|t - 1)]
\end{equation}
When the system is under attack, the state equation and its estimation are as follow
\begin{equation}
\begin{array}{l}
{x^a}(t + 1) = {\rm A}{x^a}(t) + {\rm B}{u^a}(t) + B^af(t) + {w^a}(t)
\end{array}
\end{equation}
\begin{equation}\label{eq{18}}
{{\hat x}^a}(t + 1|t) = A{{\hat x}^a}(t) + B{u^a}(t)
\end{equation}
\begin{equation}\label{eq{19}}
{{\hat x}^a}(t) = {{\hat x}^a}(t|t - 1) + L[{y^a}(t) - C{{\hat x}^a}(t|t - 1)]
\end{equation}
By subtracting (\ref{eq{15}}) from (\ref{eq{18}}), we get
\begin{equation}\label{eq{20}}
\begin{array}{l}
{{\hat x}^a}(t + 1|t) - {{\hat x}^\tau }(t + 1|t) = {A}({{\hat x}^a}(t) - {{\hat x}^\tau }(t)) + \\
 + {B}({u^a}(t) - {u^\tau }(t)) = (A + BK)(I - LC)({{\hat x}^a}(t|t - 1)\\
 - {{\hat x}^\tau }(t|t - 1)) + (A + BK)L({y^a}(t) - {y^\tau }(t))
\end{array}
\end{equation}
Since ${y^a}(t) = {y^\tau }(t)$, (\ref{eq{20}}) can be written as
\begin{equation}
\begin{array}{l}
{{\hat x}^a}(t + 1|t) - {{\hat x}^\tau }(t + 1|t) = ({A} + BK)(I - LC)\\
({{\hat x}^a}(t|t - 1) - {{\hat x}^\tau }(t|t - 1))
\end{array}
\end{equation}
Let define $\Gamma  \buildrel \Delta \over = ({A} + BK)(I - LC)$. We know that ${A}+BK$ and $A(I-LC)$ are Schur individually, but the stability of $({A} + BK)(I - LC)$ cannot be guaranteed. The residue signal of the system under attack can be obtained as
\begin{equation}\label{eq{22}}
\begin{array}{l}
{r^a}(t) = {y^a}(t) - C{{\hat x}^a}(t|t - 1) \\ ={y^\tau }(t) - C({{\hat x}^\tau }(t|t - 1) + \Gamma {}^t\left( {{{\hat x}^a}(0|- 1) - {{\hat x}^\tau }(0|-1)} \right)\,) \\ = {r^\tau }(t) - 
C\Gamma {}^t\left( {{{\hat x}^a}(0| - 1) - {{\hat x}^\tau }(0|-1)} \right)\,
\end{array}
\end{equation}
 If $\Gamma$ is Schur, then the second term in (\ref{eq{22}}) will vanish for large enough $t$, leading to converges of ${r^a}(t)$  to ${r^\tau}(t)$. Thus, we get
\begin{equation}\label{eq{66}}
\begin{array}{l}
\mathop {\lim }\limits_{t \to \infty } P({g^a}(t) > \eta ) = \mathop {\lim }\limits_{t \to \infty } P({r^a}^T(t){{H}^{ - 1}}{r^a}(t) > \eta ) \\
 = \mathop {\lim }\limits_{t \to \infty } P({r^\tau }^T(t){{H}^{ - 1}}{r^\tau }(t) > \eta ) = \alpha 
\end{array}
\end{equation}
According to (\ref{eq{66}}), in the case the matrix $\Gamma$ is stable, the available transient time for convergence of $r^a (t)$ to $r^\tau (t)$ is not reliable for detection of the attack. Thus, the number of alarms occurred in a window time during the attack can converge quickly to the number of alarms in the normal condition of the system. As a result, it is probable that attack remains stealthy in the transient time and this motivates the need for designing an attack detector method, when $\Gamma$ is stable. However, if $\Gamma$ is unstable, then both ${\hat x^a}$ and ${r^a}(t)$ diverge to infinity and  it is straightforward to show that the attack will get detected.

Note that from (12) to (22), both the state estimation and anomaly detector units use counterfeit output measurement and not the real output of the system. Consequently, $f(t)$ has no effect on the state estimation and anomaly detector units. 

\section{ATTACK DETECTOR SIGNAL DYNAMICS}
In the previous section, we studied that in the case matrix $\Gamma$ is stable, the replay attack remains undetected. 
Here, we propose a method for detection of the replay 
attack which causes the residue signal to become unstable. Using this method in a finite time, the most achievable number of alarms with arbitrary size can be obtained for a given window time $\mathscr{T}$. Moreover, the detection rate can be controlled by the defender and the design parameters can be adjusted to minimize the performance loss caused by including the attack detector signal.\\ 
Let us assume that using the accessible information of model and input-output data of the system, defender generates attack detector signal $\xi$ according to the following process
\begin{align}\label{eq{24}}
\zeta (t + 1) =& \tilde {\rm A}\zeta (t) + \tilde {\rm M}(y(t) - C\hat x(t))\\ \nonumber
\xi (t) =& \tilde {\rm K}\zeta (t)
\end{align}
where the matrices $\tilde {\rm A} \in {\mathbb{R}^{n \times n}}$, $\tilde {\rm M} \in {\mathbb{R}^{n \times p}}$ and $\tilde {\rm K} \in {\mathbb{R}^{m \times n}}$ are determined by the defender as design parameters. Fig. 1 shows a schematic of under attack control system with the proposed attack detector generation unit. Moreover, we assume that the defender uses the following input signal that is composed of control signal and generated attack detector signal
\begin{equation}\label{eq{2222}}
u(t) = K\hat x(t) + \xi (t)
\end{equation}
After injecting the attack detector signal into the system, the first requirement is that the stability of the free of attack system is guaranteed.
\begin{remark}
We assume that the system uses standard LQG controller to regulate the origin in an optimal manner. Therefore, the gain of control signal and observer parameters cannot be
changed for the sake of attack detection.
\end{remark}

\begin{remark}
The attacker can cancel out the effect of attack detector signal if he or she knows the exact model of the actual system and the structure of attack detector dynamic. Moreover, it should have access to input-output measurement, and the estimated state of the system to construct the negative of this signal. However, clearly having access to such information is not straightforward for the attacker. Moreover, we assume that the  controller unit is secure and the attacker is not able to make the attack detector signal zero.
\end{remark}
\begin{figure}[ht] 
\centerline{\includegraphics[width=8cm]{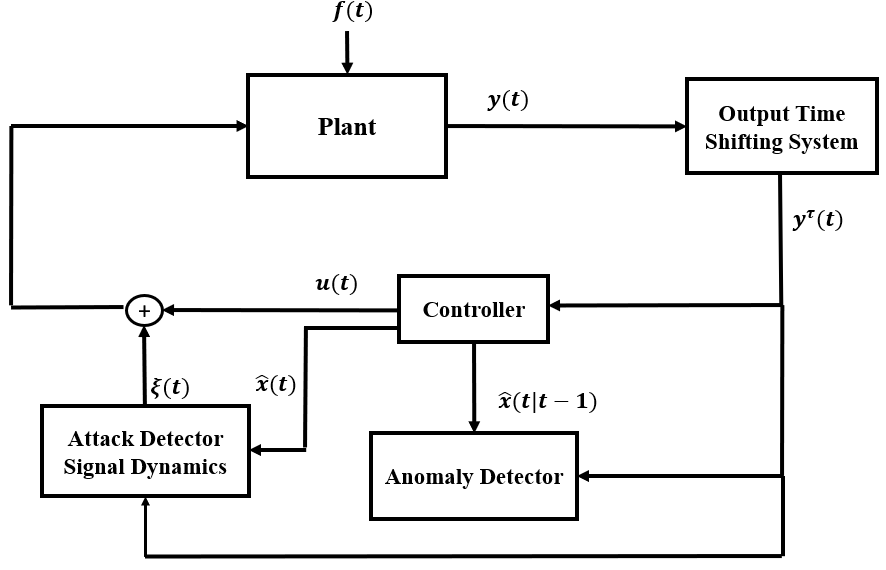}}
\caption{\small A schematic of under attack control system with a dynamical attack detector signal generation unit. Attacker replays the previous output measurements into the system. Output of the system and estimated state are used for generating attack detector signal.}
\end{figure}

\begin{Proposition}
The stability of the system (\ref{eq(1)}) after applying the control input (\ref{eq{2222}}) can be guaranteed, if and only if the matrix  $\tilde {\rm A}$ is Schur.
\end{Proposition}
\begin{proof}
The stability of the system can be verified by augmenting the dynamic of the system, estimator error and attack detector signal as follows 
\begin{equation}\label{eq{26}}
\tilde x(t + 1) = \Theta \,\tilde x(t) + \Phi \,(t)
\end{equation}     
where $e(t)=x(t)-\hat{x}(t)$ and  \\

$\begin{array}{l}
\tilde x(t) \buildrel \Delta \over = \left[ {\begin{array}{*{20}{c}}
{x(t)}\\
{\zeta (t)}\\
{e(t)}
\end{array}} \right],\\
\Theta  \buildrel \Delta \over = \,\left[ {\begin{array}{*{20}{c}}
{A + BK}&{B\tilde {\rm K}}&{ - BK}\\
0&{\tilde {\rm A}}&{\tilde {\rm M}C}\\
0&0&{(I - LC)A}
\end{array}} \right],\\
\Phi (t) \buildrel \Delta \over = \,\left[ {\begin{array}{*{20}{c}}
{w(t)}\\
{\tilde {\rm M}\upsilon (t)}\\
{(I - LC)w(t) - L\upsilon (t + 1)}
\end{array}} \right]\,
\end{array}$\\

Since the matrix $\Theta $ is upper triangular, $A + BK$ and $(I - LC)A$ are Schur therefore, the matrix $\Theta$ is Schur if and only if the matrix $\tilde{\rm A}$ is Schur. 
\end{proof}

\begin{Proposition}
	Injection of the attack detector signal has no effect on the residue signal of the system.
\end{Proposition}
\begin{proof}
	We have assumed the control signal as $u(t)=K\hat{x}(t)+\xi(t)$. Now, let's define $\bar{e}(t)\buildrel \Delta \over =x(t)-\hat x(t|t - 1)$. Then
	\begin{align}
	\bar e(t + 1) &= x(t + 1) - \hat x(t + 1|t) \nonumber\\ 
	&=Ax(t) + Bu(t)+ w(t) - A\hat x(t) - Bu(t) \nonumber\\
	&=A(x-\hat{x}(t))+w(t) \nonumber\\
	&=A\left( {x(t) - \left( {I - LC} \right)\hat x(t|t - 1) - Ly(t)} \right) + w(t) \nonumber\\
	&=A\left ( {I - LC} \right)\bar e(t) + ALv(t) + w(t)
	\end{align}
	and also we have
	\begin{align}
	r(t)&= y(t)-C\hat{x}(t|t-1)= \nonumber\\ 
	&=C\left(x(t)-\hat{x}(t|t-1)\right)+v(t)=C\bar{e}(t)+v(t)
	\end{align}

	Since $A(I-LC)$ is Schur, the expectation of $r(t)$ will be equal to zero in steady state and also its covariance will be bounded. It can be seen that the attack detector signal in control input has no effect on the residue signal in the absence of attack because $u(t)$ is a known signal and it disappears in the subtraction of $\hat{x}(t+1)$ from $x(t+1)$ 
	
\end{proof}

\subsection{Detection of replay attack using the proposed attack detector method}
In this section, we want to analyze the problem of attack detection using the proposed dynamic attack detector. State estimation dynamic of the under attack system after injection of the attack detector signal is obtained as follows
\begin{align}\label{eq{27}}
{{\hat x}^a}(t + 1|t) =& {\rm A}{{\hat x}^a}(t) + {\rm B}(K{{\hat x}^a}(t) + {\xi ^a}(t)) \\ \nonumber
 =& (A + BK)\left[ {(I - LC){{\hat x}^a}(t|t - 1) + L{y^a}(t)} \right] \\ \nonumber
&+B\tilde {\rm K}{\zeta ^a}(t)
\end{align}
Moreover, the state estimation of the attack-free system is given by
\begin{align}\label{eq{28}}
{{\hat x}^\tau }(t + 1|t)=&{\rm A}{{\hat x}^\tau }(t) + {\rm B}(K{{\hat x}^\tau }(t) + {\xi ^\tau }(t)) \\ \nonumber
 =& (A + BK)\left[ {(I - LC){{\hat x}^\tau }(t|t - 1) + L{y^\tau }(t)} \right] \\\nonumber
  &+B\tilde {\rm K}{\zeta ^\tau }(t)
\end{align}
Rewriting (\ref{eq{24}}) for both under attack and attack free conditions, the difference of the new equations is as follows
\begin{equation}\label{eq{29}}
\begin{array}{l}
{\zeta ^a}(t + 1) - {\zeta ^\tau }(t + 1) = \tilde {\rm A}({\zeta ^a}(t) - {\zeta ^\tau }(t)) \\
-\tilde {\rm M}C(I - LC)\left[ {{{\hat x}^a}(t|t - 1) - {{\hat x}^\tau }(t|t - 1))} \right]
\end{array}
\end{equation}
Also subtracting (\ref{eq{28}}) from (\ref{eq{27}}) we get
\begin{equation}\label{eq{30}}
\begin{array}{l}
{{\hat x}^a}(t + 1|t) - {{\hat x}^\tau }(t + 1|t) = \Gamma ({{\hat x}^a}(t|t - 1) - {{\hat x}^\tau }(t|t - 1)) \\
 + B\tilde {\rm K}({\zeta ^a}(t) - {\zeta ^\tau }(t))
\end{array}
\end{equation}
where we used the fact that ${y^a}(t) = {y^\tau }(t)$. By defining 
\begin{equation}\label{eq{31}}
\sigma (t) = \left[ {\begin{array}{*{20}{c}}
{{\sigma _1}(t)}\\
{{\sigma _2}(t)}
\end{array}} \right] \buildrel \Delta \over = \left[ {\begin{array}{*{20}{c}}
{{{\hat x}^a}(t|t - 1) - {{\hat x}^\tau }(t|t - 1)}\\
{{\zeta ^a}(t) - {\zeta ^\tau }(t)}
\end{array}} \right]
\end{equation}
 (\ref{eq{29}}) and  (\ref{eq{30}}) can be represented as the following state space model
\begin{equation}\label{eq{32}}
\sigma (t + 1) = \underbrace {\left[ {\begin{array}{*{20}{c}}
\Gamma &{B\tilde {\rm K}}\\
{ - \tilde {\rm M}C(I - LC)}&{\tilde {\rm A}}
\end{array}} \right]}_\Delta \sigma (t)
\end{equation}
in which $\Delta$ stands for the attack detection matrix. In what follows, we study the condition that replay attack using the proposed attack detector signal can be detected.
\begin{theorem}
If the matrix $\Delta$ has an eigenvalue with modulus greater than one, then the attack will be detected.
\end{theorem}
\begin{proof}
the residue signal of the under attack system is
\begin{equation}
\begin{split}
{r^a}(t) =& {y^a}(t) - C{{\hat x}^a}(t\left| {t - 1} \right.) = {y^\tau }(t) \\
&- C({\sigma _1}(t) + {{\hat x}^\tau }(t\left| {t - 1} \right))= {r^\tau }(t) - C{\sigma _1}(t)
\end{split}
\end{equation}
Rewriting (\ref{eq{12}}) for under attack condition we get
\begin{equation}\label{eq{34}}
\begin{split}
{g^a}(t) =& {r^a}^T(t){H^{ - 1}}{r^a}(t) = {\left( {{r^\tau }(t) - C{\sigma _1}(t)} \right)^T}{H^{ - 1}} \\
 &\times \left( {{r^\tau }(t) - C{\sigma _1}(t)} \right) = {r^\tau }^T(t){H^{ - 1}}{r^\tau }(t) \\
 &- \sigma _1^T(t){C^T}{H^{ - 1}}{r^\tau }(t) - {r^\tau }^T(t){H^{ - 1}}C{\sigma _1}(t) \\
&+ \sigma _1^T(t){C^T}{H^{ - 1}}C{\sigma _1}(t)
\end{split}
\end{equation}
It is straightforward to show that for unstable $\sigma _1(t)$ we have $\mathop {\lim }\limits_{t \to \infty } P({g^a}(t) > \eta )=1 $ and therefore, the attack will be detected. Now we show that if the matrix $\Delta$ has at least one eigenvalue outside the unit circle, then ${\sigma _1}(t)$ will be unstable. Assuming the matrix $\Delta$ has $n$ distinct eigenvalues, the solution of the system (\ref{eq{32}}) can be obtained as
\begin{equation}
\sigma (t) = {c_1}\lambda _1^t{p_1} + ... + {c_n}\lambda _n^t{p_n}
\end{equation}
where ${p_i}$ is the $i$-th right eigenvector of $\Delta$ and ${c_i}$ is a constant equal to the inner product of the $i$-th left eigenvector and initial condition $\sigma (0)$. Assume that ${\lambda _i}$ is an eigenvalue outside the unit circle. It is clear that if at least one of the upper $n$ rows of the vector ${p_i}$ is nonzero, then ${\sigma _1}(t)$ will be unstable (note that since $\sigma (0)$ is taken from a stochastic distribution, ${c_i}$ is almost surely a nonzero constant). Now, assume for the sake of contradiction the case that the $n$ upper  rows of the vector ${p_i}$ are zero. According to eigenvalue definition for ${\lambda _i}$ we have
\begin{equation}
\left[ {\begin{array}{*{20}{c}}
\Gamma &{B\tilde {\rm K}}\\
{ - \tilde {\rm M}C(I - LC)}&{\tilde {\rm A}}
\end{array}} \right]\left[ {\begin{array}{*{20}{c}}
0\\
{{p_{i,2}}}
\end{array}} \right] = {\lambda _i}\left[ {\begin{array}{*{20}{c}}
0\\
{{p_{i,2}}}
\end{array}} \right]
\end{equation}
where ${p_{i,2}}$ represents the lower $n$ rows of the vector ${p_i}$. From the above equation we have
\begin{equation}\label{eq{38}}
\tilde {\rm A}{p_{i,2}} = {\lambda _i}{p_{i,2}}
\end{equation}
From (\ref{eq{38}}), it can be observed that ${\lambda _i}$ and ${p_{i,2}}$ are the eigenvalue and eigenvector associated with the matrix $\tilde {\rm A}$, respectively. But it contradicts to the fact that all the eigenvalues of $\tilde {\rm A}$ are inside the unit circle. Therefore, we can conclude that the upper $n$ rows of eigenvector ${p_i}$  associated with the unstable eigenvalue cannot be all equal to zero.
\end{proof}

In [17], the detectability of attack has been verified when the matrix $\Gamma=(A+BK)(I-LC)$ is unstable. In our method, by adding the attack detector dynamic as expanding the dynamics of the underlying system, the matrix $\Gamma$ is extended to matrix $\Delta$. In addition to increasing probability of attack detection in the proposed dynamic attack detection method, the extended matrix $\Delta$ contains some designed parameters which can be used to optimize the performance of detection method.

\begin{remark}
Even if we have an eigenvalue with some multiplicity, there exists at least one independent eigenvector associated with that eigenvalue.  Therefore, the discussion on the independent eigenvectors which is presented in proof of Theorem 1 can be applied since, we do not need to use the other generalized eigenvectors for this purpose and the proof remains unchanged.
\end{remark}

\begin{remark}
In theorem 1, from $\mathop {\lim }\limits_{t \to \infty } P({g}(t) > \eta )=1$, it can be deduced that in infinite time the anomaly detector gives the most achievable number of alarms $\beta_{max}$ (continues alarm) in a given window time, nevertheless, it is sufficient that the number of alarms exceeds a threshold to detect the attack. Thus, if this threshold is less than $\beta_{max}$, the attack can be revealed in a finite time.
\end{remark}

In the following theorem we show that according to the degrees of freedom on choosing the matrices $\tilde {\rm M}$, $\tilde {\rm K}$ and $\tilde {\rm A}$, the defender can make the matrix $\Delta$ unstable.
\begin{theorem}
If the matrix $\Gamma$ is full rank, by adjusting the matrices $\tilde {\rm M}$, $\tilde {\rm K}$ and stable $\tilde {\rm A}$, the matrix $\Delta$ will have at least an eigenvalue with modulus greater than one.
\end{theorem}

\begin{proof}
We know that the product of all eigenvalues of a matrix is equal to the determinant of that matrix. If $\det (\Delta ) > 1$, then there exists at least one eigenvalue of $\Delta$ that is outside of the unit circle. In what follows, we prove that if the matrix $\Gamma$ is full rank, by adjusting the design matrices $\tilde {\rm M}$, $\tilde {\rm K}$ and $\tilde {\rm A}$ one can make the determinant of the matrix $\Delta$ greater than one. We know that if $\det (\Gamma ) \ne 0$, we have
\begin{align}
\begin{array}{l}
\det \left( {\left[ {\begin{array}{*{20}{c}}
\Gamma &{B\tilde {\rm K}}\\
{-\tilde {\rm M}C(I - LC)}&{\tilde {\rm A}}
\end{array}} \right]} \right) = \det \left( \Gamma  \right) \\
 \times \det \left( {\tilde {\rm A} + \tilde {\rm M}C(I - LC){\Gamma ^{ - 1}}B\tilde {\rm K}} \right)
\end{array}
\end{align}
Since $\Gamma$ is Schur, the absolute of its determinant is lower than one. Assume that the element $(i,j)$ of $\tilde {\rm M}$ and the element $(r,s)$ of $\tilde {\rm K}$ get two arbitrary values $m$ and $k$, respectively, and the other elements of $\tilde {\rm M}$ and $\tilde {\rm K}$ are set to zero. It is obvious that only $i$-th row of $\tilde {\rm M}C(I - LC)$ and $s$-th column of $B\tilde {\rm K}$ remain nonzero. Therefore, the matrix $\tilde {\rm M}C(I - LC){\Gamma ^{ - 1}}B\tilde {\rm K}$ will have only one nonzero element in its $(i,s)$ entry, that can be chosen in a desired way. The determinant of $\tilde{\rm A}$ along the $i$-th row appears as a linear combination of the determinants of the minors $M_{ij}$, which cannot be all equal to zero if $\tilde{\rm A}$ is chosen an invertible matrix. Let ${a_{is}}$ be the nonzero element associated with such a nonzero determinant, therefore we have
\begin{align}
\det (\tilde{\rm A} + \tilde{\rm M}C(I - LC){\Gamma ^{ - 1}}B\tilde K) = {\mu _1} \\ \nonumber
+{\mu _2}mk\det ({M_{is}})
\end{align}
where $\mu _1$ and $\mu _2$ are constants. Then, one can choose $m$ and $k$ in a way that $det(\Delta)$ is greater than one. 
\end{proof}
\begin{Proposition}
The matrix $\Gamma$ is full rank if and only if the matrix $A$ is full rank.
\end{Proposition}
\begin{proof}
We know that the multiplication of two full rank matrices is a full rank matrix and vice versa. Therefore, the matrix $\Gamma$ is full rank if and only if $({A} + BK)$ and $(I - LC)$ are full rank. First we show the matrix $(I-LC)$ is always full rank. Also we know that the following equality holds for any $M \in {\mathbb{R}^{p \times k}}$ and $N \in {\mathbb{R}^{k \times p}}$ 
\begin{equation}
\det ({I_{p \times p}} - MN) = \det ({I_{k \times k}} - NM)
\end{equation}
Therefore, using the fact that $L = {\Sigma _e}{C^T}{(C{\Sigma _e}{C^T} + R)^{ - 1}}$, $\det ({I_{n \times n}} - LC)$ can be written as
\begin{equation}
\begin{array}{l}
\det \left( {{I_{n \times n}} - LC} \right) = \det \left( {{I_{p \times p}} - CL} \right) = \\
\det \left( {I - C{\Sigma _e}{C^T}{{\left( {C{\Sigma _e}{C^T} + R} \right)}^{ - 1}}} \right)=\\
\det \left( {\left( {C{\Sigma _e}{C^T} + R - C{\Sigma _e}{C^T}} \right){{\left( {C{\Sigma _e}{C^T} + R} \right)}^{ - 1}}} \right) = \\
\det \left( {R{{\left( {C{\Sigma _e}{C^T} + R} \right)}^{ - 1}}} \right) = \det \left( R \right) \\
 \times \det \left( {{{\left( {C{\Sigma _e}{C^T} + R} \right)}^{ - 1}}} \right)
\end{array}
\end{equation}
Since the matrix $R$ is full rank and positive definite, and also the matrix $C{\Sigma _e}{C^T}$ is a positive semi-definite matrix, therefore, $\det ({I_{n \times n}} - LC)>0$.\\
For the matrix $A+BK$, we can follow the same procedure. $\det (A+BK)$ can be written in the following form
\begin{equation}
\begin{split}
&\det({A + BK}) = \det \left( {A - B{{\left( {{B^T}PB + U} \right)}^{ - 1}}{B^T}PA} \right) \\
&=\det (A)\det \left( {{I_{n \times n}} - B{{\left( {{B^T}PB + U} \right)}^{ - 1}}{B^T}P} \right) \\
&=\det (A)\det \left( {{I_{m \times m}} - {B^T}PB{{\left( {{B^T}PB + U} \right)}^{ - 1}}} \right)= \\
&\det (A)\det \bigg( {\Big( {{B^T}PB + U - {B^T}PB} \Big)
{{\Big( {{B^T}PB + U} \Big)}^{ - 1}}} \bigg)\\
&=\det (A)\det \left( {U{{\left( {{B^T}PB + U} \right)}^{ - 1}}} \right) = \det (A)\det (U) \times \\
&\,\,\,\,\,\,\,\det \left( {{{\left( {{B^T}PB + U} \right)}^{ - 1}}} \right)
\end{split}
\end{equation}
Since the matrix $U$ is positive definite and also the ${B^T} PB$ is positive semi definite, therefore, $\det (U)\det \left( {{{\left( {{B^T}PB + U} \right)}^{ - 1}}} \right)>0$. As a result, $A+BK$ is full rank if and only if the matrix $A$ is full rank and the proof is completed. 
%We can conclude that for the systems with independent state space representation ($\det (A) \ne 0$), the matrix $(A+BK)(I-LC)$ is full rank. 
\end{proof}

\begin{remark}
The cyber-physical systems which are under study of this paper are continuous time in reality. It is straightforward to show that the matrix $A$ which can be obtained after discretization of any continuous time dynamic is full rank and therefore, the condition in Proposition 3 is always satisfied.
\end{remark}
Now, among all the feasible values for design matrices $\tilde {\rm M}$, $\tilde {\rm K}$ and stable $\tilde {\rm A}$, we seek for those with minimum
performance loss effect in the absence of attack. In what follows we obtain the performance loss of the
system after injection of the proposed attack detector signal.
\subsection{Performance Loss}
Here, we follow the procedure for obtaining the performance loss (new cost function) after injection of the attack detector signal.\\ 
The input of the system after inserting attack detector signal is as follows
\begin{equation}
u(t) = K\hat x(t) + \xi (t) = K\hat x(t) + \tilde {\rm K}\zeta (t)
\end{equation}
The performance loss is defined as
\begin{equation}\label{eq{4224}}
\tilde J = \mathop {\lim }\limits_{N \to \infty } \frac{1}{N}\sum\limits_{t = 1}^{N} {E\left\{ {{x^T}(t)Wx(t) + {u^T}(t)Uu(t)} \right\}}  
\end{equation}
Therefore, after injection of the attack detector signal to the system, the performance loss can be written as
\begin{equation}
\begin{split}
\tilde J =& \mathop {\lim }\limits_{N \to \infty } \frac{1}{N}\sum\limits_{t = 1}^{N} E\bigg\{ {x^T}(t)Wx(t) + \Big(K\big(x(t) - e(t)\big) \\
 &+ \tilde {\rm K}\zeta (t)\Big)^TU \Big(K\big(x(t) - e(t)\big) + \tilde {\rm K}\zeta (t)\Big)\bigg\}  =\\
=& \mathop {\lim }\limits_{N \to \infty } \frac{1}{N}\sum\limits_{t = 1}^{N} {E\bigg\{ {{\left[ {\begin{array}{*{20}{c}}
 {{x^T}(t)}&{{\zeta ^T}(t)}&{{e^T}(t)}
 \end{array}} \right]}}}\\
 =&\underbrace{\begin{bmatrix}
     W + {K^T}UK & {{K^T}U\tilde {\rm K}} & { - {K^T}UK}\\
     {\tilde {\rm K}^TUK} & {\tilde {\rm K}^TU\tilde {\rm K}} & { - \tilde {\rm K}^TUK}\\{ - {K^T}UK}&{ - {K^T}U\tilde {\rm K}}&{{K^T}UK}
   \end{bmatrix}}_G 
   \begin{bmatrix}
   {x(t)}\\
   {\zeta (t)}\\
   {e(t)}
   \end{bmatrix}\bigg\}
\end{split}
\end{equation}
\begin{equation}
\begin{split}
\tilde J =& \mathop {\lim }\limits_{N \to \infty } \frac{1}{N}\sum\limits_{t = 1}^{N} {E\left\{ {tr\left( {{{\tilde x}^T}(t)G\tilde x(t)} \right)} \right\}} \\
=&\mathop {\lim }\limits_{N \to \infty } \frac{1}{N}\sum\limits_{t = 1}^{N} {E\left\{ {tr\left( {\tilde x(t){{\tilde x}^T}(t)G} \right)} \right\}} \\
=&\mathop {\lim }\limits_{N \to \infty } \frac{1}{N}\sum\limits_{t = 1}^{N} {tr\left( {G{\mathop{\rm cov}} (\tilde x(t))} \right)}  = tr\left( {G{\mathop{\rm cov}} (\tilde x)} \right)\,
\end{split}
\end{equation}
where according to (\ref{eq{26}}), $cov(\tilde{x})$ can be obtained by solving the following Lyapunov equation
\begin{equation}
{\mathop{\rm cov}} (\tilde x) = \Theta {\mathop{\rm cov}} (\tilde x){\Theta ^T} + E\{\Phi \Phi {\,^T} \}
\end{equation}
By taking the expectation values of the elements of $\Phi {\,^T}\Phi $ we get
\begin{equation}\label{eq{44}}
\scalebox{0.95}{$\begin{array}{l}
{\mathop{\rm cov}} (\tilde x) = \Theta {\mathop{\rm cov}} (\tilde x){\Theta ^T} + \\
\underbrace {\left[ {\begin{array}{*{20}{c}}
Q&0&{Q{{(I - LC)}^T}}\\
0&{\tilde {\rm M}R{{\tilde {\rm M}}^T}}&0\\
{(I - LC)Q}&0&{(I - LC)Q{{(I - LC)}^T} + LR{L^T}}
\end{array}} \right]}_\Psi 
\end{array}$}
\end{equation}
% The following optimization problem considers the abovementioned conditions.
\subsection{Optimization Problem }
The optimal values for $\tilde {\rm M}$, $\tilde {\rm K}$ and $\tilde {\rm A}$ can be found from the following optimization problem
\begin{equation}\label{eq{45}}
\begin{array}{l}
\mathop {\min }\limits_{\tilde {\rm A},\,\,\tilde {\rm M},\,\tilde {\rm K}} tr\left( {G{\mathop{\rm cov}} (\tilde x)} \right)\\
st.\,\,\,\,{\mathop{\rm cov}} (\tilde x) = \Theta {\mathop{\rm cov}} (\tilde x){\Theta ^T} + \Psi \\
1 < \delta  \le \max \{ abs(eig(\Delta ))\} \\
\tilde {\rm A}:Schur\,\,and\,\,full\,rank
\end{array}
\end{equation}
\begin{Proposition}
The optimization problem (\ref{eq{45}}) is feasible.
\end{Proposition}
\begin{proof}
Since the matrix $\Theta$ is Schur, according to Proposition 1 and also due to the fact that the matrix $\Psi$ in the right hand side of (\ref{eq{44}}) is a non-negative definite matrix, there exists a unique non-negative definite solution to the Lyapunov equation (\ref{eq{44}}). On the other hand, in Theorem 2, we showed that one can adjust the absolute value of the maximum eigenvalue of the matrix $\Delta$ to become greater than any scalar $\delta >1$. Therefore, the feasible domain of the optimization problem (\ref{eq{45}}) is non-empty.
\end{proof}
It should be mentioned that the optimization problem (\ref{eq{45}}) is a non-convex nonlinear programming.  Nevertheless, there are many nonlinear programming methods for solving non-convex problems such as interior point, trust-region-reflective, etc. 

\section{SIMULATION RESULTS}
In this section, DC motor systems that are commonly used in cyber-physical systems are considered to verify the results of the paper. By assuming that the magnetic field is constant, the state space model of the system is LTI and can be represented as follows
\begin{equation}\label{eq{1276}}
\frac{d}{{dt}}\left[ {\begin{array}{*{20}{c}}
\theta \\
{\dot \theta }\\
i
\end{array}} \right] = \left[ {\begin{array}{*{20}{c}}
0&1&0\\
0&{ - \frac{b}{J}}&{\frac{K_t}{J}}\\
0&{ - \frac{K_b}{L}}&{ - \frac{R}{L}}
\end{array}} \right]\left[ {\begin{array}{*{20}{c}}
\theta \\
{\dot \theta }\\
i
\end{array}} \right] + \left[ {\begin{array}{*{20}{c}}
0\\
0\\
{\frac{1}{L}}
\end{array}} \right]V(t)
\end{equation}
\begin{equation}
y(t) = \left[ {\begin{array}{*{20}{c}}
1&0&0
\end{array}} \right]\left[ {\begin{array}{*{20}{c}}
\theta \\
{\dot \theta }\\
i
\end{array}} \right]
\end{equation}

The physical parameters for this example are given in Table 1. It is assumed that the model (\ref{eq{1276}}) is discretized with sampling time $T_s=.01s$. The cost matrices for this system, W and U, and covariance matrices, Q and R, are equal to the identity matrix with appropriate dimensions. The cost of LQG controller is $j^*=1.04$ and we set the scalar $\gamma$ in a way that $\alpha=.01$. We fix the window time $\mathscr{T}=1s$ (equal to 100 sample time). If the number of alarms occurred in this window time exceed the desired number of alarms $\beta$, then the attack will be revealed.  

It is assumed that the attacker records output measurements from time interval $[-15 \,\ 0]s$ and replays them exactly in time interval $[0 \,\ 15]s$.  
For this system, the matrix $\Gamma$ has three eigenvalues inside the unit circle. Consequently, the attack will remain stealthy and the defender has to use another countermeasure against the attack.
 
We have designed our proposed attack detector signal by solving the optimization problem (\ref{eq{45}}) for $\delta=1.03$. The design matrices are as follows
\begin{equation}
\begin{split}
\tilde {\rm A} &= \left[ {\begin{array}{*{20}{c}}
{.48}&{ - .81}&{.02}\\
{.01}&{.61}&{ - .92}\\
{.89}&{.73}&{ - .9}
\end{array}} \right],\quad \tilde {\rm M} = \left[ {\begin{array}{*{20}{c}}
{ - .84}\\
{ - .49}\\
{ - .81}
\end{array}} \right] \\
 \tilde {\rm K} &= \left[ {\begin{array}{*{20}{c}}
 {.9}&{ - .1}&{.35}
 \end{array}} \right]
  \end{split}
  \end{equation}
Since the system behaves stochastically, we use expected 
detection time (the average time that is needed
 to reach a desired number of alarms $\beta$)
 instead of detection time. In Fig. 2, we have compared the expected detection time in our proposed method with white Gaussian watermarking signal method in the case that performance loss is equal for both methods. Each individual curve is obtained by averaging on 1000 experiments.\\
In Fig. 2, it can be seen that our proposed method gives
lower expected detection time for all desired number of
alarms. However, for $\beta\ge5$ the White Gaussian is not capable of detecting the presence of attack but the proposed attack detector method can even reveal the attack for the most achievable number of alarms $\beta_{max}=100$.

\begin{figure}[ht]
\centerline{\includegraphics[width=8.5cm]{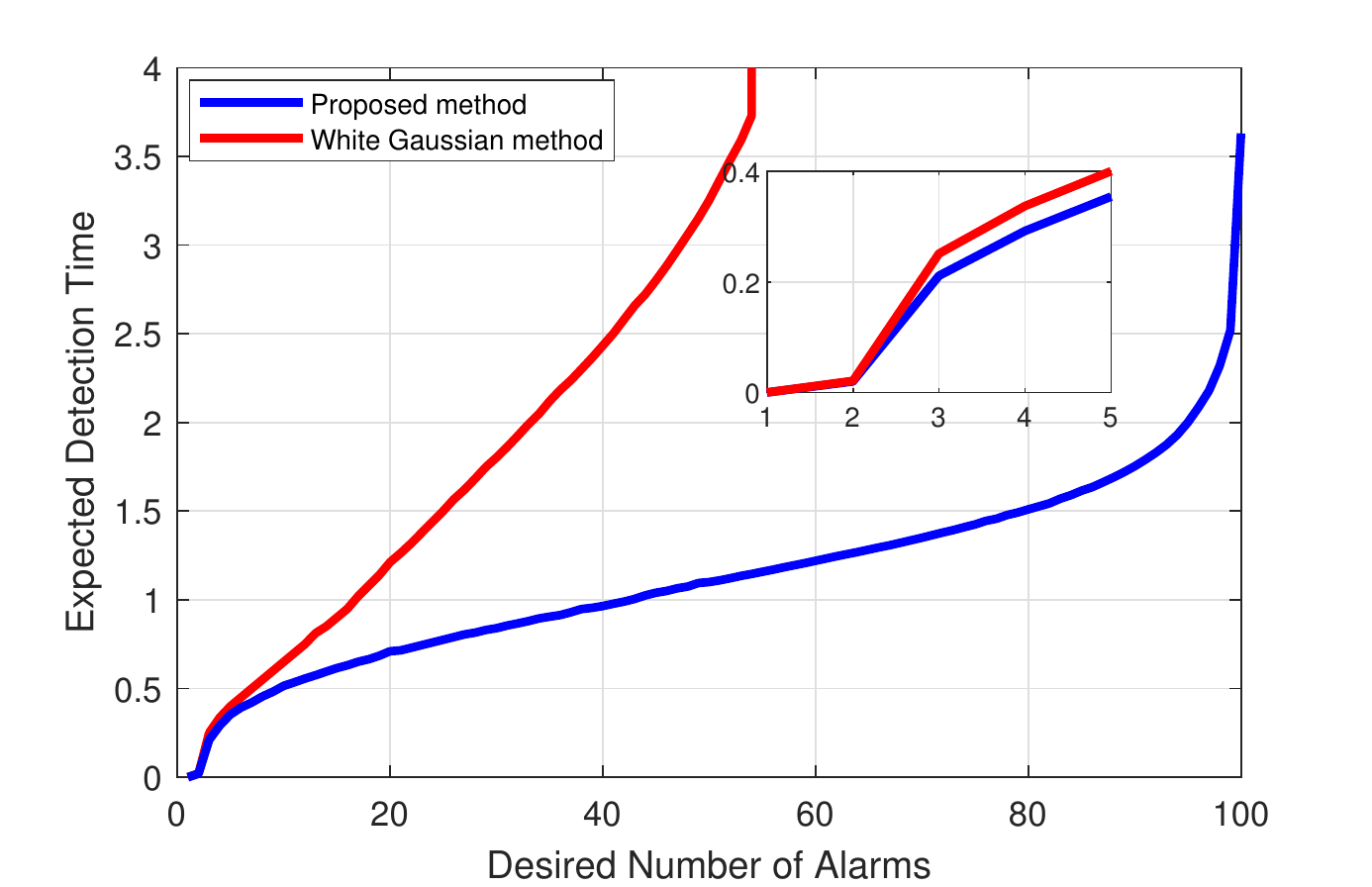}}
\caption{\small Expected detection time versus number of alarms for our proposed method with $\delta=1.03$ and white Gaussian watermarking signal for the same performance loss.}
\end{figure}

\begin{table}[h!]
  \begin{center}
    \caption{Physical parameters of DC motor}
    \label{tab:table1}
    \begin{tabular}{c|c} % <-- Alignments: 1st column left, 2nd middle and 3rd right, with vertical lines in between
      Parameter & Value \\
      \hline
      
      $b$ & $3.5 \times {10^{ - 6}}\,\, N.m.s$\\
      $J$ & $3.23 \times {10^{-6}}\,\, kg.m^2$\\
      $L$ & 2.75 $\times {10^{-6}}\,\, H$\\
      $K_b$ & $0.0274\,\, V/rad/sec $\\
      $K_t$ & $0.0274\,\,  N.m/Amp$ \\
      $R$ &  $4\,\,\, Ohm$
    \end{tabular}
  \end{center}
\end{table}

\begin{figure}[ht]
\centerline{\includegraphics[width=8.5cm]{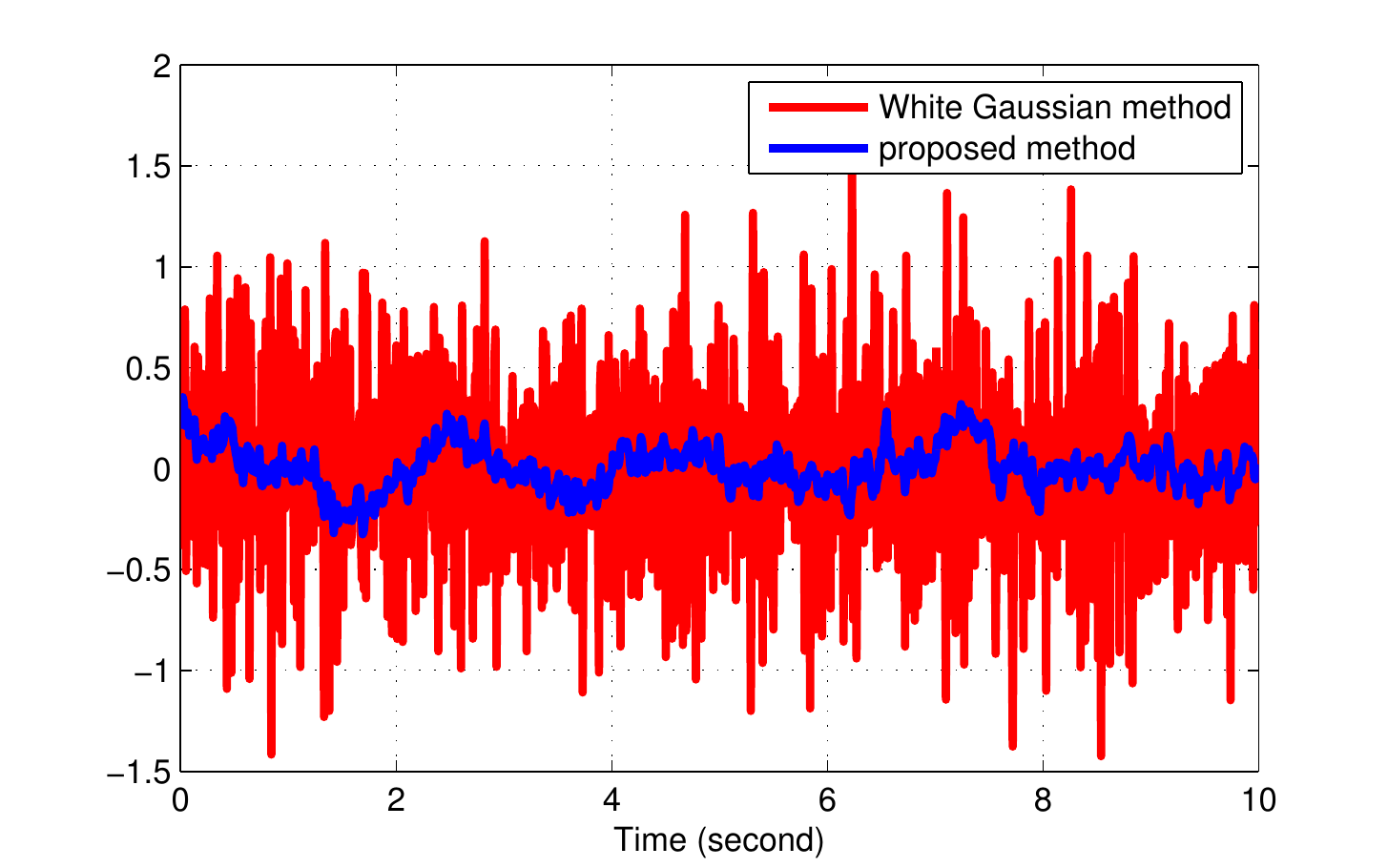}}
\caption{\small Control signal after injection of our proposed (blue line) and watermarking signal (red line).}
\end{figure}

\begin{figure}[ht]
\centerline{\includegraphics[width=8.5cm]{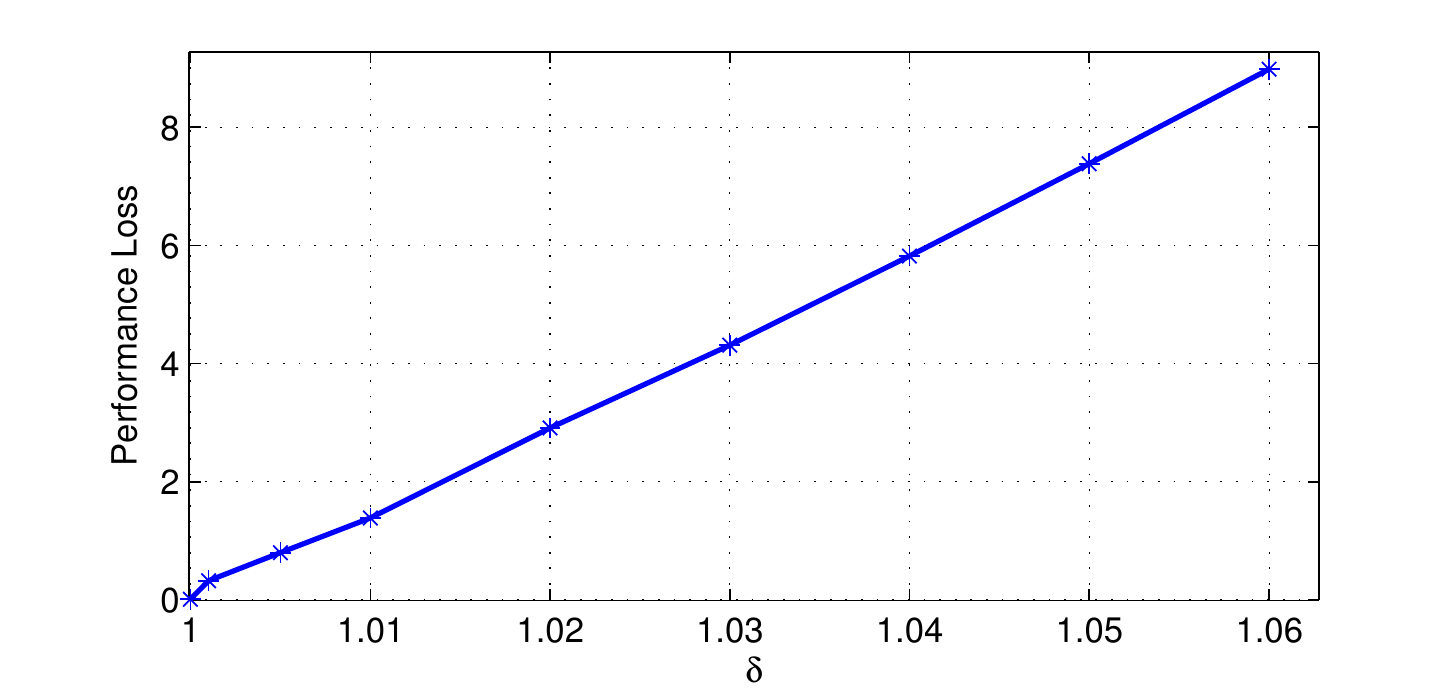}}
\caption{\small  The amount of performance loss versus $\delta$. By increasing $\delta$ the amount of performance loss will also increase. }
\end{figure}

\newcounter{tempEquationCount} 
\setcounter{MYtempeqncnt}{\value{equation}}\addtocounter{equation}{1}

In Fig. 3 we compare the control signal of our proposed method with control signal of white Gaussian watermarking method in the case that they give the same detection time ($.5s$) for the same desired number of alarms $\beta=30$. As it can be seen, the control signal in our proposed attack detector is smoother, has smaller range and lower energy with respect to the control signal in white Gaussian noise method.

Moreover, Fig. 4 examines the relation between the amount of performance loss and maximum eigenvalue of detection matrix $\Delta$. It is clear that by increasing $\delta$ the amount of performance loss  increases as well.


\begin{thebibliography}{99}

\bibitem{c1} Neuman, Clifford. "Challenges in security for cyber-physical systems." DHS Workshop on Future Directions in Cyber-Physical Systems Security. 2009.


\bibitem{c2} Langner, Ralph. "To kill a centrifuge: A technical analysis of what stuxnet's creators tried to achieve." The Langner Group (2013).

\bibitem{c3} S. Karnouskos, "Stuxnet worm impact on industrial cyber-physical
system security," in Proc. 37th Annu. Conf. IEEE Industrial Electronics
Society, Melbourne, Australia, 2011, pp. 4490-4494.

\bibitem{c4} Liu, Yao, Peng Ning, and Michael K. Reiter. "False data injection attacks against state estimation in electric power grids." ACM Transactions on Information and System Security (TISSEC) 14.1 (2011): 13.
\bibitem{c5} Pasqualetti, Fabio, Florian Dorfler, and Francesco Bullo. "Attack detection and identification in cyber-physical systems." IEEE Transactions on Automatic Control 58.11 (2013): 2715-2729.
\bibitem{c6} Teixeira, Andre, et al. "A secure control framework for resource-limited adversaries." Automatica 51 (2015): 135-148.
 
\bibitem{c7} Smith, Roy S. "Covert misappropriation of networked control systems: Presenting a feedback structure." IEEE Control Systems 35.1 (2015): 82-92.

\bibitem{c8} Yuan, Yuan, et al. "Resilient control of cyber-physical systems against denial-of-service attacks." 2013 6th International Symposium on Resilient Control Systems (ISRCS). IEEE, 2013.

\bibitem{c9} Pang, Zhong-Hua, and Guo-Ping Liu. "Design and implementation of secure networked predictive control systems under deception attacks." IEEE Transactions on Control Systems Technology 20.5 (2011): 1334-1342.

\bibitem{c10} Ding, Derui, et al. "Event-based security control for discrete-time stochastic systems." IET Control Theory \& Applications 10.15 (2016): 1808-1815.

\bibitem{c11} Fawzi, Hamza, Paulo Tabuada, and Suhas Diggavi. "Secure estimation and control for cyber-physical systems under adversarial attacks." IEEE Transactions on Automatic Control 59.6 (2014): 1454-1467.

\bibitem{c12} Shoukry, Yasser, and Paulo Tabuada. "Event-triggered state observers for sparse sensor noise/attacks." IEEE Transactions on Automatic Control 61.8 (2016): 2079-2091.

\bibitem{c13} Shoukry, Yasser, et al. "SMT-based observer design for cyber-physical systems under sensor attacks." ACM Transactions on Cyber-Physical Systems 2.1 (2018): 5.

\bibitem{c14} Pajic, Miroslav, Insup Lee, and George J. Pappas. "Attack-resilient state estimation for noisy dynamical systems." IEEE Transactions on Control of Network Systems 4.1 (2017): 82-92.

\bibitem{c15}Nateghi, Shamila, et al. "Cyber Attack Reconstruction of Nonlinear Systems via Higher-Order Sliding-Mode Observer and Sparse Recovery Algorithm." 2018 IEEE Conference on Decision and Control (CDC). IEEE, 2018.

\bibitem{c16} Mo, Yilin, and Bruno Sinopoli. "Secure control against replay attacks."Communication, Control, and Computing, 2009. Allerton 2009. 47th Annual Allerton Conference on. IEEE, 2009.
\bibitem{c17} Mo, Yilin, Rohan Chabukswar, and Bruno Sinopoli. "Detecting integrity attacks on SCADA systems." IEEE Transactions on Control Systems Technology 22.4 (2014): 1396-1407.
\bibitem{c18} Mo, Yilin, Sean Weerakkody, and Bruno Sinopoli. "Physical authentication of control systems: designing watermarked control inputs to detect counterfeit sensor outputs." IEEE Control Systems 35.1 (2015): 93-109.

\bibitem{c19} Liu, Hanxiao, et al. "An On-line Design of Physical Watermarks." 2018 IEEE Conference on Decision and Control (CDC). IEEE, 2018.

\bibitem{c20} Tang, Bixiang, Luis D. Alvergue, and Guoxiang Gu. "Secure networked control systems against replay attacks without injecting authentication noise." American Control Conference (ACC), 2015. IEEE, 2015.
\bibitem{c21} Hoehn, Andreas, and Ping Zhang. "Detection of 
replay attacks in cyber-physical systems." American Control Conference (ACC), 2016. IEEE, 2016.

\bibitem{c22} Khazraei, Amir, Hamed Kebriaei, and Farzad Rajaei Salmasi. "Replay attack detection in a multi agent system using stability analysis and loss effective watermarking." American Control Conference (ACC), 2017. IEEE, 2017.

\bibitem{c23} Khazraei, Amir, Hamed Kebriaei, and Farzad Rajaei Salmasi. "A new watermarking approach for replay attack detection in LQG systems." Decision and Control (CDC), 2017 IEEE 56th Annual Conference on. IEEE, 2017.

 
\bibitem{c24} Simon, Dan. Optimal state estimation: Kalman, H infinity, and nonlinear approaches. John Wiley \& Sons, 2006.

\bibitem{c25} Willsky, Alan S. "A survey of design methods for failure detection in dynamic systems." Automatica 12.6 (1976): 601-611.





\end{thebibliography}
\end{document}